\documentclass[12pt]{article}

\usepackage[hmargin=1.3in,vmargin=1.3in]{geometry}

\usepackage{amsfonts,amssymb,amsmath,amsthm,bm}
\usepackage{color}

\newtheorem{theorem}{Theorem}

\newtheorem{remark}{Remark}
\newtheorem{corollary}{Corollary}
\newtheorem{lemma}{Lemma}
\newtheorem{proposition}{Proposition}
\newtheorem{conjecture}{Conjecture}

\begin{document}

\title{On deep holes of Gabidulin codes }
\author{\small  Weijun Fang$^{a}$\ \ Li-Ping Wang$^{b,c,}$\thanks{Corresponding Author}\ \ Daqing Wan$^d$\\
\small $^a$  Chern Institute of Mathematics and LPMC, Nankai University, Tianjin 300071, China\\
\small Email: nankaifwj@163.com\\
\small $^b$ Institute of Information Engineering, Chinese Academy of Sciences,
Beijing 100093, China\\
 \small $^{c}$ University of Chinese Academy of Sciences,  Beijing, China \\
\small Email: wangliping@iie.ac.cn\\
\small $^d$ Department of Mathematics,
University of California, Irvine
CA 92697, USA\\ \small Email: dwan@math.uci.edu}
\date{}
\maketitle
\thispagestyle{empty}

\begin{abstract}
 In this paper,  we study  deep holes of  Gabidulin codes in both rank  and Hamming metrics.   Specifically, first,  we give a tight  lower bound for the distance of any word  to  a Gabidulin code  and  a sufficient and necessary condition for achieving this lower bound as well.  Then,  a class of deep holes of   a Gabidulin code  are  discovered.  Furthermore, we obtain some other deep holes for certain   Gabidulin codes.
\end{abstract}


\maketitle

\section{Introduction}

 Let  $\mathbb{F}^{n}_{q^m}$ be an $n$-dimensional vector space over a finite field
$\mathbb{F}_{q^m}$ where  $q$ is a prime power,  and  $n, m$ are positive  integers.  In this paper we only consider the case when $n\leq m$.  Let
$\mathbf{\beta}=(\beta_1,\ldots,\beta_{m})$ be a basis of $\mathbb{F}_{q^m}$ over
$\mathbb{F}_q$.  Let $\mathcal{F}_i$ be the map from $\mathbb{F}_{q^m}$ to $\mathbb{F}_q$ where $\mathcal{F}_i(u)$ is the $i$-th coordinate of  an element  $u\in \mathbb{F}_{q^m}$ in the basis representation with $\mathbf{\beta}$.   To any $\mathbf{u}=(u_1,\ldots, u_n)$ in $\mathbb{F}^{n}_{q^m}$,
we may associate the matrix $\bar{\mathbf{u}}=(\bar{u}_{i, j})_{1\leq i\leq m, 1\leq j\leq n}\in \mathcal{M}_{m,n}(\mathbb{F}_q)$ in which $\bar{u}_{i,j}=\mathcal{F}_i(u_j)$.  The rank weight of  the vector $\mathbf{u}$ can be defined by the rank of the associated matrix $\bar{\mathbf{u}}$, denoted by $w_{R}(\mathbf{u})$.  Thus, we can define the rank distance between two vectors $\mathbf{u}$ and $\mathbf{v}$ in $\mathbb{F}^{n}_{q^m}$ as $d_{R}(\mathbf{u},\mathbf{v})=w_{R}(\mathbf{u}-\mathbf{v})$.
 We refer to \cite{Loid}
for more details on codes for the rank distance.

For integers $1\leq k\leq n$, a linear rank-metric code $C$ of length $n$ and dimension $k$ over $\mathbb{F}_{q^m}$  is a subspace
of dimension $k$  of $\mathbb{F}^{n}_{q^m}$ embedded with the  rank metric.   The minimum rank distance  of the code $C$, denoted by $d_{R}(C)$,   is the minimum rank weight of the non-zero codewords in $C$. A linear rank-metric code $C$ of length $n$ and dimension $k$ over $\mathbb{F}_{q^m}$ is called a maximum rank distance (MRD) code if $d_{R}(C)=n-k+1$.
A $k\times n$  matrix is called a generator matrix of $C$ if its rows span the code.

 The rank distance of any word $\mathbf{u}\in \mathbb{F}^{n}_{q^m}$ to $C$ is defined as
   \[d_{R}(\mathbf{u}, C)=\min\{d_R(\mathbf{u}, \mathbf{c})\mid \mathbf{c}\in C\}.\]  It plays an important role in  decoding of rank-metric codes.  The maximum rank distance
\[\rho_R(C)=\max\{d_{R}(\mathbf{u}, C) \mid \mathbf{u}\in \mathbb{F}^{n}_{q^m}\}\]
is called the covering radius of $C$.  If the rank distance from a word to the code $C$ achieves the covering radius of the code, the word is called  a deep hole of the code $C$.

 The covering radius and deep holes of a linear code  embedded with  Hamming metric were studied extensively \cite{Bartoli, Qi, QiWan,Cohen85, Cohen86,graham, Tor, Keti, Liao, Wan2008, Wu, Zhang2013, ZW2016, ZW2017,  Zhuang}, in which MDS codes such as generalized Reed-Solomon codes, standard Reed-Solomon codes and projective Reed-Solomon codes were explored deeply.
  Gabidulin codes were   introduced by  Gabidulin in \cite{gabidulin} and independently by Delsarte in \cite{del}. Gabidulin codes can be seen as the $q$-analog of Reed-Solomon codes. Furthermore,
 Gabidulin codes are MRD codes.  Over the last decade there has been increased interest in Gabidulin codes, mainly because of their relevance to network coding \cite{Kotter, Silva}.  The covering radius for a Gabidulin code  was also  studied in \cite{Yan1, Yan2, Vasa}.
 However, little is known about  deep holes for such a code. In this paper, we give a tight lower bound for the distance of any word to a Gabidulin code in both rank  and Hamming metrics,  and a sufficient and necessary condition for attaining this lower bound  as well.  Then,  a class of deep holes of a   Gabidulin code  are discovered.  Furthermore,  we study the distance of a special class of words to  a Gabidulin code  and so obtain some other deep holes for certain   Gabidulin codes. Note that we refer to  rank metric if  Hamming metric is not explicitly pointed out in this paper.

The rest of this paper is organized as follows. In Section 2, we introduce some basic notations and results about linearized polynomials.  Section 3 provides a class of deep holes for a  Gabidulin code in both rank and Hamming metrics.  Next,  we obtain some other deep holes for certain   Gabidulin codes in Section 4.  Finally, we give  our conclusions  in Section 5.

\section{Linearized polynomials}

Gabidulin codes exploit linearized polynomials instead of arbitrary polynomials and so we recall some results about linearized polynomials.

A $q$-linearized polynomial over $\mathbb{F}_{q^m}$ is defined to be a polynomial of the form
\[L(x)=\sum_{i=0}^{d}a_i x^{q^i}, a_i\in\mathbb{F}_{q^m},  a_d\neq 0\]
where $d$ is called the $q$-degree of $f(x)$, denoted by $\deg_q(f(x))$.  Note that $L(x)$ has no constant term.
One can easily check that $L(x_1+x_2)=L(x_1)+L(x_2)$ and $L(\lambda x_1)=\lambda L(x_1)$ for any
$x_1, x_2\in \mathbb{F}_{q^m}$ and $\lambda\in \mathbb{F}_q$, from  which the name stems.
In particular, $L(x)$ induces an $\mathbb{F}_q$-linear endomorphism of the $\mathbb{F}_q$-vector
space $\mathbb{F}_{q^m}$.
The set of all $q$-linearized polynomials over $\mathbb{F}_{q^m}$ is denoted by  $\mathcal{L}_{q}(x,\mathbb{F}_{q^m})$. The ordinary product of linearized polynomials does not have to be a linearized polynomial. However, the composition $L_1(x)\circ L_2(x)=L_1(L_2(x))$ is also a linearized polynomial. The set $\mathcal{L}_{q}(x,\mathbb{F}_{q^m})$ forms a non-commutative ring under the operations of composition $\circ$ and ordinary addition. It is also an $\mathbb{F}_q$-algebra.

\medskip

\begin{lemma}\cite{Nied}
Let  $f(x)\in \mathcal{L}_{q}(x,\mathbb{F}_{q^m})$ and $\mathbb{F}_{q^s}$  be the smallest extension field of $\mathbb{F}_{q^m}$ that contains all roots of $f(x)$. Then the set of all roots of $f(x)$ forms an $\mathbb{F}_q$-linear vector space in $\mathbb{F}_{q^s}$.
\end{lemma}

\medskip
Let $U$ be an $\mathbb{F}_q$-linear subspace of $\mathbb{F}_{q^m}$. Then $\prod_{g\in U}(x-g)$ is called the $q$-annihilator polynomial of $U$.

\medskip
\begin{lemma}\cite{Nied}
Let $U$ be an $\mathbb{F}_q$-linear subspace of $\mathbb{F}_{q^m}$. Then $\prod_{g\in U}(x-g)$ is a $q$-linearized polynomial over $\mathbb{F}_{q^m}$.
\end{lemma}

\medskip

Let  $\beta_1,\ldots, \beta_n\in \mathbb{F}_{q^m}$ and  denote the $k\times n$
Moore matrix by
\[M_{k}(\beta_1,\ldots,\beta_n):=\left ( \begin{array}{cccc}
\beta_1 & \beta_2 & \ldots & \beta_{n} \\
\beta_{1}^{q} & \beta_{2}^{q} & \ldots & \beta_{n}^{q} \\
\vdots & \vdots & \ddots & \vdots \\
\beta_{1}^{q^{k-1}} & \beta_{2}^{q^{k-1}} & \ldots & \beta_{n}^{q^{k-1}}
\end{array}
\right ).
\]
Furthermore, if $g_1,\ldots, g_n$ is a basis of $U$, one can write
\[\prod_{g\in U}(x-g)=\lambda \det(M_{n+1}(g_1,\ldots,g_n,x))\]
for some non-zero constant $\lambda\in \mathbb{F}_{q^m}$.  Clearly, its $q$-degree is $n$.

In addition, we have the notion of $q$-Lagrange polynomials.

Let $\mathbf{g}=\{g_1,\ldots,g_n\}\subset \mathbb{F}_{q^m}$ and $\mathbf{r}=\{r_1,\ldots,r_n\}\subset \mathbb{F}_{q^m}$,  where $g_1,\ldots,g_n$ are $\mathbb{F}_{q}$-linearly independent.
For $1\leq i \leq n$, we define the matrix $\mathcal{D}_i(\mathbf{g},x)$ as $M_n(g_1,\ldots, g_n, x)$ without the $i$th column. The $q$-Lagrange
polynomial with respect to $\mathbf{g}$ and $\mathbf{r}$ is defined to be
\[\Lambda_{\mathbf{g},\mathbf{r}}(x)=\sum_{i=1}^{n}(-1)^{n-i}r_i \frac{\det(\mathcal{D}_i(\mathbf{g},x))}{\det(M_n(\mathbf{g}))}\in \mathbb{F}_{q^m}[x].   \]

\begin{proposition}\cite{Wachter}\label{prop: Lagrange}
 The $q$-Lagrange polynomial $\Lambda_{\mathbf{g},\mathbf{r}}(x)$ is a $q$-linearized polynomial in $\mathbb{F}_{q^m}[x]$ and $\Lambda_{\mathbf{g},\mathbf{r}}(g_i)=r_i$ for $i=1,\ldots,n$.
\end{proposition}

\begin{proposition}\cite{Kuijper} \label{prop:multi}
Let $L(x)\in \mathcal{L}_q(x, \mathbb{F}_{q^m})$ be such that $L(g_i)=0$ for all $i$. Then there exists an $H(x)\in \mathcal{L}_q(x, \mathbb{F}_{q^m})$ such that
$L(x)=H(x) \circ \prod_{g\in <\mathbf{g}>}(x-g)$, where  $<\mathbf{g}>$ is the $\mathbb{F}_q$-vector space spanned by $\mathbf{g}$. \end{proposition}

\section{Deep holes of Gabidulin codes}

Let $g_1,\ldots, g_n\in \mathbb{F}_{q^m}$ be linearly independent over $\mathbb{F}_q$, which also implies that $n\leq m$.  Let $\mathbf{g}=\{g_1,\ldots,g_n\}$ and $<\mathbf{g}>$ is the $\mathbb{F}_q$-vector space spanned by $\mathbf{g}$.   A Gabidulin code $\mathcal{G}\subseteq \mathbb{F}_{q^m}^n$ is defined as a  linear block code with the generator matrix $M_k(g_1,\ldots,g_n)$, where $1\leq k\leq n$. Using the isomorphic matrix representation, we can interpret $\mathcal{G}$ as a matrix code in $\mathbb{F}_{q}^{m\times n}$.  The rank distance is defined in Section 1.

The Gabidulin code $\mathcal{G}$ with length $n$ has dimension $k$ over $\mathbb{F}_{q^m}$ and minimum rank distance $n-k+1$, and so $\mathcal{G}$ is an MRD code \cite{gabidulin}.  The Gabidulin code $\mathcal{G}$ can also be defined as follow:
\begin{eqnarray} \label{def}
\mathcal{G} &= & \{(m(g_1),\ldots, m(g_n))\in \mathbb{F}^{n}_{q^m}|\, m(x)\in \mathcal{L}_q(x, \mathbb{F}_{q^m})   \nonumber \\
 & & \mbox{ and } \deg_q(m(x))<k \}.
 \end{eqnarray}
Note that this interpretation of the code $\mathcal{G}$ will be used throughout the rest of the paper.
It is the $q$-analogue of the generalized Reed-Solomon code.

Let
$$(\prod_{g\in <\mathbf{g}>}(x-g))=\mathcal{L}_q(x,\mathbb{F}_{q^m})\circ \prod_{g\in <\mathbf{g}>}(x-g)$$
be the left ideal generated by the element $\prod_{g\in <\mathbf{g}>}(x-g)$ in the non-commutative ring
$\mathcal{L}_q(x,\mathbb{F}_{q^m})$ with respect to the composition product.
In particular,  $(\prod_{g\in <\mathbf{g}>}(x-g))$ is an $\mathbb{F}_q$-linear additive subgroup of   $\mathcal{L}_q(x,\mathbb{F}_{q^m})$.
It follows that $\mathcal{L}_q(x,\mathbb{F}_{q^m})/(\prod_{g\in <\mathbf{g}>}(x-g))$ is an $\mathbf{F}_q$-vector space.  Define an $\mathbb{F}_q$-linear evaluation map
$$\sigma: \mathcal{L}_q(x,\mathbb{F}_{q^m})/(\prod_{g\in <\mathbf{g}>}(x-g)) \longrightarrow \mathbb{F}^{n}_{q^m}$$
given  by
$$\sigma(f(x))=(f(g_1),\ldots, f(g_n)).$$
We have the following property.

 \medskip
 \begin{proposition} \label{prop: map}
 The  above defined map $\sigma$ is an $\mathbb{F}_q$-vector space isomorphism.
 \end{proposition}

\begin{proof} First,  $\sigma$ is well-defined since   the polynomial $\prod_{g\in <\mathbf{g}>}(x-g)$ vanishes at every $g_i$.
Second, if $f(g_i)=0$ for all $i=1,\ldots, n$, then there exists $H(x)\in \mathcal{L}_q(x,\mathbb{F}_{q^m})$ such that $f(x)=H(x)\circ \prod_{g\in <\mathbf{g}>}(x-g)$ by Proposition  \ref{prop:multi} and so $\sigma$ is one-to-one.  Third, we show that $\sigma$ is surjective.  For a given $\mathbf{r}=(r_1,\ldots,r_n)\in \mathbb{F}^{n}_{q^m}$,  we have the
$q$-Lagrange polynomial $\Lambda_{\mathbf{g},\mathbf{r}}(x)$ satisfying  $\Lambda_{\mathbf{g},\mathbf{r}}(g_i)=r_i$ for $i=1,\ldots,n$ by Proposition \ref{prop: Lagrange}.  The result is proved.
\end{proof}

\smallskip
The $q$-linearized polynomial $\prod_{g\in <\mathbf{g}>}(x-g)$ has $q$-degree $n$. It follows that any element $f(x)\in \mathcal{L}_q(x,\mathbb{F}_{q^m})$ can be written uniquely
in the form
$$f(x)=h(x)\circ \prod_{g\in <\mathbf{g}>}(x-g)   + r(x),$$
where $h(x), r(x) \in \mathcal{L}_q(x,\mathbb{F}_{q^m})$ and $r(x)$ has $q$-degree smaller than $n$. This is the $q$-division algorithm
in the non-commutative ring $\mathcal{L}_q(x,\mathbb{F}_{q^m})$. As $\mathbb{F}_q$-vector spaces, the quotient
$\mathcal{L}_q(x,\mathbb{F}^{m}_q)/(\prod_{g\in <\mathbf{g}>}(x-g))$ is thus represented by all $q$-linearized polynomials of $q$-degree less than $n$. That is,
$$\mathcal{L}_q(x,\mathbb{F}_{q^m})/(\prod_{g\in <\mathbf{g}>}(x-g)) =\{ f \in \mathcal{L}_q(x,\mathbb{F}_{q^m}) |\deg_q(f)<n \}.$$
Using the isomorphism $\sigma$,  we can identify any word $u\in \mathbb{F}^n_{q^m}$ with $\sigma(f)$ for
a unique polynomial $f(x)\in \mathcal{L}_q(x,\mathbb{F}_{q^m})$ with $\deg_q(f)< n$. When $\deg_q(f)\leq k-1$, it is easy to see that the distance $d_R(\sigma_f, \mathcal{G})=0$ by the definition.  It was proved in \cite{Yan2} that the covering radius of $\mathcal{G}$ is $n-k$.  Thus, we have $d_R(\sigma_f, \mathcal{G}) \leq n-k$ by the definition of covering radius. When $k\leq \deg_q(f)<n$, we provide a  tight lower bound for
$d_R(\sigma_f, \mathcal{G})$ as follows.

\medskip
\begin{theorem} \label{thm: property}
Let $f(x)\in \mathcal{L}_q(x,\mathbb{F}_{q^m})$ with $\deg_q(f)< n$ and let $\sigma_f = \sigma(f) \in \mathbb{F}^n_{q^m}$ be the corresponding word. If $k\leq \deg_q(f)<n$, then
$$d_R(\sigma_f, \mathcal{G}) \geq n-\deg_q (f).$$
Furthermore, we suppose $f$ is monic, then
$d_R(\sigma_f, \mathcal{G})=n-\deg_q(f)$ if and only if there exists a $\deg_q(f)$-dimensional subspace $H$ of $<\mathbf{g}>$ such that
\[f(x)-v(x)=\prod_{h\in H}(x-h),\]
for some $v(x)\in \mathcal{L}_q(x, \mathbb{F}_{q^m})$ with $\deg_q(v)\leq k-1$.
\end{theorem}

\medskip
\begin{proof}
Let $u(x)$ be any $q$-polynomial over $\mathbb{F}_{q^{m}}$. We consider the $\mathbb{F}_q$-linear map defined by
\begin{eqnarray*}
  \pi_{u} &:&  <g_{1}, \cdots, g_{n}> \rightarrow  <u(g_{1}), \cdots, u(g_{n}) >\\
   & & \sum_{i=1}^{n}\xi_{i}g_{i} \mapsto \sum_{i=1}^{n}\xi_{i} u(g_{i})=u(\sum_{i=1}^{n}\xi_{i}g_{i}).
\end{eqnarray*}
It is clear that the map $\pi_{u}$ is surjective and ker$(\pi_{u}) \subseteq $ Root$(u)$ (the set of roots of $u(x)$). So dim$_{\mathbb{F}_{q}}$ ker$(\pi_{u}) \leq $dim$_{\mathbb{F}_{q}}$ Root$(u)\leq$ $\deg_{q}(u)$. Then
\begin{eqnarray*}
   & & \mbox{dim}_{\mathbb{F}_q} <u(g_{1}), \cdots, u(g_{n}) > \\
  &=&\mbox{dim}_{\mathbb{F}_q} <g_{1}, \cdots, g_{n} > - \mbox{dim}_{\mathbb{F}_q} \mbox{ker}(\pi_{u})\\
    &\geq& n-\deg_{q}(u).
\end{eqnarray*}
It follows that
\begin{eqnarray*}
   & & d_R(\sigma_f, \mathcal{G})\\
   &=&\min\limits_{\deg_{q}(v) < k}\textnormal{rank}((f-v)(g_{1}),\cdots, (f-v)(g_{n}))\\
   &=&\min\limits_{\deg_{q}(v) < k}\textnormal{dim}_{\mathbb{F}_q}<(f-v)(g_{1}), \cdots, (f-v)(g_n)> \\
    &\geq& \min\limits_{\deg_{q}(v) < k}(n-\deg_{q}(f-v))=n-\deg_{q}(f).
\end{eqnarray*}
The last equality holds since $\deg_{q}(f-v)=\deg_{q}(f)$ for any $q$-polynomial $v(x)$ with $\deg_{q}(v) < k$.

Furthermore, from the above proof, we know  $d_R(\sigma_f, \mathcal{G})=n-\deg_q(f)$ if and only if
\begin{eqnarray*}
   \textnormal{dim}_{\mathbb{F}_{q}}\textnormal{Root}(f-v) &=& \textnormal{dim}_{\mathbb{F}_{q}} \textnormal{ker}(\pi_{f-v}) \\
   &=& \deg_{q}(f-v)=\deg_{q}(f)
\end{eqnarray*}
for some $q$-polynomial $v(x)$ with $\deg_{q}(v) < k$, which is equivalent to
$$f(x)-v(x)=\prod_{h \in H}(x-h),$$
for some $\deg_{q}(f)$-dimensional subspace $H$ of $<g_{1}, \cdots, g_{n} >$.
The theorem is proved.
\end{proof}

By Theorem \ref{thm: property} and the fact $d_R(\sigma_f, \mathcal{G}) \leq n-k$, we immediately deduce the following corollary, which provide a class of deep holes of the Gabidulin code $\mathcal{G}$.

\medskip

\begin{corollary} \label{cor: deepholes}
The elements of the set $\{\sigma_f: \deg_q(f(x))=k, f(x)\in  \mathcal{L}_q(x,\mathbb{F}_{q^m})\}$ are  deep holes of the Gabidulin code $\mathcal{G}$ and so the number of deep holes of $\mathcal{G}$ is at least  $(q^{m}-1)q^{mk}$.
\end{corollary}

\medskip

According to the definition in Eq. (\ref{def}), we may also study  Gabidulin codes in  Hamming metric.  It was showed that such codes are MDS codes in \cite{gabidulin}. We use $d_{H}(\mathbf{u},\mathbf{v})$ and $d_H(\mathbf{u}, \mathcal{G})$ to denote the Hamming distance between vectors $\mathbf{u}$ and  $\mathbf{v}$ and the Hamming distance of a word $\mathbf{u}$ to $\mathcal{G}$, respectively.  Similarly, we have the following theorem.

\begin{theorem} \label{thm: propertyH}
Let $f(x)\in \mathcal{L}_q(x,\mathbb{F}_{q^m})$ with $\deg_q(f)< n$ and let $\sigma_f = \sigma(f) \in \mathbb{F}^n_{q^m}$ be the corresponding word. If $k\leq \deg_q(f)<n$, then
$$d_{H}(\sigma_f, \mathcal{G})\geq n-\deg_q (f).$$
Furthermore, suppose $f$ is monic, then
$d_H(\sigma_f, \mathcal{G})=n-\deg_q(f)$ if and only if there exists a subset $E=\{g_{i_1},\ldots,
g_{i_{\deg_q(f)}}\}$ of $\{g_1, \cdots, g_n\}$ such that
\[f(x)-v(x)=\prod_{g\in <E>}(x-g),\]
for some $v(x)\in \mathcal{L}_q(x, \mathbb{F}_{q^m})$ with $\deg_q(v)\leq k-1$.
\end{theorem}

\medskip
\begin{proof}
Let now $t= n-d_H(\sigma_f, \mathcal{G})$.
By definition of the Hamming distance, there exists  some $v(x)$ and non-zero $H(x) \in \mathcal{L}_q(x,\mathbb{F}_{q^m})$ with $\deg_q(v)<k$ such that
\[f(x)-v(x)=H(x)\circ \prod_{g\in <g_{i_1}, \ldots,g_{i_t}>}(x-g)\]
for some indices $1\leq i_1<\cdots < i_t \leq n$.
Comparing the $q$-degrees of both sides, we deduce that  $t \leq \deg_q(f)$. This proves that   $n-\deg_q (f)\leq d_H(\sigma_f, \mathcal{G})$.
Furthermore, if $f$ is monic, the equality $t=\deg_q(f)$ holds if and only if $H(x)=x$, in which case, we obtain
\[f(x)-v(x)=\prod_{g\in <g_{i_1}, \ldots,g_{i_t}>}(x-g)\]
and the theorem is true.
\end{proof}

It is well known that $d_H(\mathbf{u}, C) \leq n-k$ for any linear code of length $n$ and dimension $k$. Thus, by Theorem \ref{thm: propertyH},  the result in Corollary \ref{cor: deepholes} still holds in Hamming metric.

\section{Some other deep holes for  certain  Gabidulin codes}

We hope to obtain more deep holes of Gabidulin codes and so consider monic $f(x)$ of $\deg_{q}(f)=k+d, d \geq 1$, where  $f(x)=x^{q^{k+d}}-a_{1}x^{q^{k+d-1}}+a_{2}x^{q^{k+d-2}}+\cdots+(-1)^{d}a_{d}x^{q^{k}}+\cdots$.
In Theorem \ref{thm: property}, if we write $\prod_{h \in H}(x-h)=x^{q^{k+d}}-h_{1}x^{q^{k+d-1}}+\cdots+(-1)^{d}h_{d}x^{q^{k}}+\cdots$ and let $\beta_{1}, \beta_{2}, \ldots, \beta_{k+d} \in \mathbb{F}_{q^{m}}$ be a basis of $H$, then $d_R(\sigma_f, \mathcal{G})=n-\deg_{q}(f)$ is equivalent to
$$a_{i}=h_{i},\textnormal{ for all }1 \leq i \leq d. $$
According to the process of the proof of \cite[Lemma 3.51]{Nied}, we know that
\[h_{i}=\frac{\det(\mathcal{R}_{k+d-i}(\beta_1,\cdots,\beta_{k+d}))
}{\det(M_{k+d}(\beta_1,\cdots,\beta_{k+d}))},\]
where
  $\mathcal{R}_{k+d-i}(\beta_1,\cdots,\beta_{k+d})$ denotes the matrix $M_{k+d+1}(\beta_1,\cdots$,  $\beta_{k+d})$  deleting the row $(\beta_{1}^{q^{k+d-i}}, \cdots, \beta_{k+d}^{q^{k+d-i}})$. As a result,  we have
 $d_R(\sigma_f, \mathcal{G})=n-(k+d)$ if and only if there exist $k+d$ linearly independent elements $\beta_{1}, \beta_{2}, \ldots, \beta_{k+d}$ of $<g_{1}, \cdots, g_{n}>$ such that
$$a_{i}=\frac{\det(\mathcal{R}_{k+d-i}(\beta_1,\cdots,\beta_{k+d}))
}{\det(M_{k+d}(\beta_1,\cdots,\beta_{k+d}))},
              \mbox{ for all } 1 \leq i \leq d,$$
  where  $\mathcal{R}_{k+d-i}(\beta_1,\cdots,\beta_{k+d})$ denotes as the above.

\medskip

When $d=1$, i.e., $\deg_q(f)=k+1$, then by  Theorem \ref{thm: property}, $\sigma_{f}$ is not a deep hole of $\mathcal{G}$ if and only if $d_R(\sigma_f, \mathcal{G})=n-(k+1)$.  Thus, by the above discussion, we have

\medskip
\begin{lemma}\label{lem: deg of k+1}
Let $f(x)=x^{q^{k+1}}-a_{1}x^{q^{k}}+\cdots$.  Then $\sigma_{f}$ is not a deep hole of $\mathcal{G}$ if and only if there exist $k+1$ linearly independent elements $\beta_{1}, \beta_{2}, \ldots, \beta_{k+1}$ of $<g_{1}, \cdots, g_{n}>$ such that
$$a_{1}=\frac{\det(\mathcal{R}_{k}(\beta_1, \cdots,\beta_{k+1}))}{\det(M_{k+1}(\beta_1,\cdots,\beta_{k+1}))},$$
where $\mathcal{R}_{k}(\beta_1, \cdots,\beta_{k+1})$ denotes the matrix $M_{k+2}(\beta_1$, $\cdots$, $\beta_{k+1})$ without the row $(\beta_{1}^{q^k}, \cdots$,  $\beta_{k+1}^{q^{k}})$.
\end{lemma}

Similar to the above discussion, we get the result for Hamming metric case by Theorem \ref{thm: propertyH}.
Let $f(x)=x^{q^{k+d}}-a_{1}x^{q^{k+d-1}}+a_{2}x^{q^{k+d-2}}+\cdots+(-1)^{d}a_{d}x^{q^{k}}+\cdots$. Then $d_H(\sigma_f, \mathcal{G})=n-(k+d)$ if and only if there exist $k+d$ distinct elements  $g_{i_1}, g_{i_2}, \ldots, g_{i_{k+d}}$ of $\{g_{1}, \cdots, g_{n}\}$ such that
$$a_{i}=\frac{\det(\mathcal{R}_{k+d-i}(g_{i_1},\cdots, g_{i_{k+d}}))
}{\det(M_{k+d}(g_{i_1},\cdots, g_{i_{k+d}}))},
              \mbox{ for all } 1 \leq i \leq d,$$
  where  $\mathcal{R}_{k+d-i}(g_{i_1},\cdots, g_{i_{k+d}})$ denotes as the above.

\medskip
\begin{lemma}\label{lem: deg of k+1Ham}
Let $f(x)=x^{q^{k+1}}-a_{1}x^{q^{k}}+\cdots$.  Then $\sigma_{f}$ is not a deep hole of $\mathcal{G}$ in  Hamming metric if and only if there exist $k+1$ distinct elements $g_{i_1}, g_{i_2}, \ldots, g_{i_{k+1}}$ of $\{g_{1}, \cdots, g_{n}\}$  such that
$$a_{1}=\frac{\det(\mathcal{R}_{k}(g_{i_1},\cdots, g_{i_{k+1}}))}{\det(M_{k+1}(g_{i_1},\cdots, g_{i_{k+d}}))},$$
where $\mathcal{R}_{k}(g_{i_1},\cdots, g_{i_{k+d}})$ denotes the matrix $M_{k+2}(g_{i_1}$, $\cdots$,  $g_{i_{k+d}})$ without the row $(g_{i_1}^{q^k}, \cdots$, $g_{i_{k+1}}^{q^{k}})$.
\end{lemma}

In the following,  we study  some other  deep holes for certain   Gabidulin codes.  In particular, we consider Gabidulin codes over $\mathbb{F}_{q^m}$ only  when $m=n$ in Proposition \ref{prop: ex1} and \ref{prop: ex2}.

\begin{proposition} \label{prop: ex1}
Let $\mathcal{G}$ be the Gabidulin code over $\mathbb{F}_{q^n}$ with linearly independent set $\mathbf{g}=\{g_1,\ldots, g_n\}$ and dimension $k$.  Let $f(x)=x^{q^{n-1}}+f_{\leq k-1}$, where $f_{\leq k-1}$ is a $q$-linearized polynomial over $\mathbb{F}_{q^n}$ of $q$-degree less than or equals to $k-1$. Then $\sigma_f$ is  a deep hole of $\mathcal{G}$.
\end{proposition}
\medskip
\begin{proof}
For any $h_{1}, h_{2}, \ldots, h_{n} \in \mathbb{F}_{q^n}$, it is easy to show that
\[\textnormal{dim}_{\mathbb{F}_q}<h_{1}, h_{2}, \cdots, h_n>=\textnormal{dim}_{\mathbb{F}_q}<h_{1}^q, h_{2}^q, \cdots, h_{n}^q>.\]
Thus we have
\begin{eqnarray*}
   & & d_R(\sigma_f, \mathcal{G})\\
   &=&\min\limits_{\deg_{q}(v) < k}\textnormal{rank}((f-v)(g_{1}),\cdots, (f-v)(g_{n}))\\
   &=&\min\limits_{\deg_{q}(v) < k}\textnormal{dim}_{\mathbb{F}_q}<(f-v)(g_{1}), \cdots, (f-v)(g_n)> \\
   &=&\min\limits_{\deg_{q}(v) < k}\textnormal{dim}_{\mathbb{F}_q}<(f-v)^q(g_{1}), \cdots, (f-v)^q(g_n)> \\
   &=&\min\limits_{\deg_{q}(v) < k}\textnormal{dim}_{\mathbb{F}_q}<g_{1}+(f_{\leq k-1}-v)^q(g_{1}), \cdots, g_n+(f_{\leq k-1}-v)^q(g_n)> \\
    &\geq& \min\limits_{\deg_{q}(v) < k}(n-\deg_{q}(x+(f_{\leq k-1}(x)-v(x))^{q}))\\
    &\geq& n-k.
\end{eqnarray*}
The fourth equality holds since $g_{i}^{q^{n}}=g_{i}$, and the first inequality follows from the process of the proof of Theorem 1. By the fact $d_R(\sigma_f, \mathcal{G}) \leq n-k$, we obtain that $d_R(\sigma_f, \mathcal{G}) = n-k$. Thus $\sigma_f$ is  a deep hole of $\mathcal{G}$.
\end{proof}

\medskip

In Proposition \ref{prop: ex1}, if the dimension $k$ equals to $n-2$, we can obtain more deep holes of the Gabidulin code as follows.
\begin{proposition} \label{prop: ex2}
Let $\mathcal{G}$ be the Gabidulin code over $\mathbb{F}_{q^n}$ with linearly independent set $\mathbf{g}=\{g_1,\ldots, g_n\}$ and dimension $k=n-2$.  Let $f(x)=x^{q^{n-1}}-ax^{q^{n-2}}+f_{\leq n-3}$, where $a$ is an element in $\mathbb{F}_{q^{n}}$ with $a \neq (-1)^{n-1}b^{1-q}$ for all $b \in \mathbb{F}^{*}_{q^{n}}$ and $f_{\leq n-3}$ is a $q$-linearized polynomial over $\mathbb{F}_{q^n}$ of $q$-degree less than or equals to $n-3$.
 Then $\sigma_f$ is  a deep hole of $\mathcal{G}$.
\end{proposition}

\medskip
\begin{proof}
Suppose that $\sigma_f$ is not a deep hole of $\mathcal{G}$. By Lemma \ref{lem: deg of k+1}, there are $n-1$ linearly independent elements $\beta_1, \ldots, \beta_{n-1}$ in $<g_1,\cdots, g_n>$ such that
\begin{equation}\label{2}
  a=\frac{\det(\mathcal{R}_{n-2}(\beta_1, \cdots,\beta_{n-1}))}{\det(M_{n-1}(\beta_1,\cdots,\beta_{n-1}))}.
\end{equation}
For any matrix $A=(a_{ij})$, denote by $A^{(q)}$ the matrix $(a^{q}_{ij})$. Then
\[\mathcal{R}^{(q)}_{n-2}(\beta_1, \cdots,\beta_{n-1})=\left ( \begin{array}{cccc}
\beta_1^{q} & \beta_2^{q} & \ldots & \beta_{n-1}^{q} \\
\beta_{1}^{q^{2}} & \beta_{2}^{q^{2}} & \ldots & \beta_{n-1}^{q^{2}} \\
\vdots & \vdots & \ddots & \vdots \\
\beta_{1}^{q^{n-2}} & \beta_{2}^{q^{n-2}} & \ldots & \beta_{n-1}^{q^{n-2}}\\
\beta_{1}^{q^{n}} & \beta_{2}^{q^{n}} & \ldots & \beta_{n-1}^{q^{n}}
\end{array}
\right ).\]
Note that $\beta_{i}^{q^{n}}=\beta_{i}$. Thus
\[\mathcal{R}^{(q)}_{n-2}(\beta_1, \cdots,\beta_{n-1})=\left ( \begin{array}{cccc}
\beta_1^{q} & \beta_2^{q} & \ldots & \beta_{n-1}^{q} \\
\beta_{1}^{q^{2}} & \beta_{2}^{q^{2}} & \ldots & \beta_{n-1}^{q^{2}} \\
\vdots & \vdots & \ddots & \vdots \\
\beta_{1}^{q^{n-2}} & \beta_{2}^{q^{n-2}} & \ldots & \beta_{n-1}^{q^{n-2}}\\
\beta_{1} & \beta_{2} & \ldots & \beta_{n-1}
\end{array}
\right ),\]
and $\det(\mathcal{R}^{(q)}_{n-2}(\beta_1, \cdots,\beta_{n-1}))=(-1)^{n-1}\det(M_{n-1}(\beta_1,\cdots,\beta_{n-1})).$
It is easy to see that $\det(A^{(q)})=(\det(A))^{q}$, for any matrix $A$ over $\mathbb{F}_{q^n}$. Thus
\begin{eqnarray*}
(\det(\mathcal{R}_{n-2}(\beta_1, \cdots,\beta_{n-1})))^{q}&= & \det(\mathcal{R}^{(q)}_{n-2}(\beta_1, \cdots,\beta_{n-1}))\\
   &= & (-1)^{n-1}\det(M_{n-1}(\beta_1,\cdots,\beta_{n-1}))\neq 0.
   \end{eqnarray*}
i.e., $\det(\mathcal{R}_{n-2}(\beta_1, \cdots,\beta_{n-1})) \neq 0$. Moreover, by Eq. (2), we have
\[a=(-1)^{n-1}(\det(\mathcal{R}_{n-2}(\beta_1, \cdots,\beta_{n-1})))^{1-q},\]
which contradicts to the assumption of $a$. Thus $\sigma_f$ is a deep hole of $\mathcal{G}$.
\end{proof}

\medskip

\begin{remark} When $a=0$,  the result in Proposition \ref{prop: ex2} can be obtained by Proposition \ref{prop: ex1}.
\end{remark}

The following proposition considers the case of Gabidulin codes with dimension $k=1$.
\begin{proposition} \label{prop: ex3}
Suppose $m$ is odd and $3\leq n\leq m$. Let $\mathcal{G}$ be the Gabidulin code with linearly independent set $\mathbf{g}=\{g_1,\ldots, g_n\}$ and dimension $k=1$.   Let $f(x)=x^{q^2}+cx$ where $ c\in \mathbb{F}_{q^m}$. Then  $\sigma_f$ is  a deep hole of $\mathcal{G}$.
\end{proposition}

\medskip

\begin{proof}
Suppose that $\sigma_f$ is not a deep hole of $\mathcal{G}$. By Lemma  \ref{lem: deg of k+1}, there are two linearly independent elements $\beta_1$ and $\beta_2$ in $<g_1,\cdots, g_n>$ such that
$b=0=\beta_1\beta_2(\beta_{2}^{q^2-1}-\beta_{1}^{q^2-1})$.
Thus, $(\beta_2\beta^{-1}_{1})^{q^2-1}=1$.  Since $m$ is odd, $\gcd(q^2-1,q^{m}-1)=q-1$.  So we have $(\beta_2\beta^{-1}_{1})^{q-1}=1$, which implies that $\frac{\beta_2}{\beta_1} \in \mathbb{F}_{q}$, i.e., $\beta_1$ and $\beta_2$ are linearly dependent over $\mathbb{F}_{q}$.  This contradicts with the assumption of $\beta_1$ and $\beta_2$.
\end{proof}

\begin{remark} When $n=m=3$, the result in Proposition \ref{prop: ex3} is included in Proposition \ref{prop: ex2}.
\end{remark}

\medskip
Propositions \ref{prop: ex1}, \ref{prop: ex2} and \ref{prop: ex3} still hold for the Hamming metric after similar analysis.

\medskip
In the rest of this section  we furthermore discuss the distance of a special class of words to the Gabidulin codes over $\mathbb{F}_{2^{m}}$ with dimension $k=1$.  Before that, we give two lemmas.

\medskip

\begin{lemma}\label{lem:trace} \cite{Nied}
Let $a$ be in a finite field $\mathbb{F}_q$ and $p$ be the characteristic of $\mathbb{F}_q$. Then  the trinomial $x^{p}-x-a$ is irreducible in $\mathbb{F}_q[x]$ if and only if
Tr$_{\mathbb{F}_q/\mathbb{F}_p}(a)\neq 0$.
\end{lemma}

\medskip

For  a finite field $\mathbb{F}_q$,  the integer-valued function $v$ on $\mathbb{F}_q$ is defined by $v(b)=-1$ for $b\in \mathbb{F}^*_{q}$ and $v(0)=q-1$.

\medskip

\begin{lemma}\label{lem:solutions} \cite{Nied}
For even $q$,   let $a\in \mathbb{F}_q$ with tr$_{\mathbb{F}_q}(a)=1$ and $b\in \mathbb{F}_q$, then
the number of solutions of the equation $x^2_1+x_1x_2+ax^2_{2}=b$ is $q-v(b)$.
\end{lemma}

\medskip

We now consider  the finite field $\mathbb{F}_{2^m}$.   Let
\[h(x_1,x_2)=x^2_1+x_1x_2+x^2_{2}. \]
For any $b\in \mathbb{F}_{2^m}$,  let the  set
\begin{equation*}
S(h(x_1,x_2)=b)= \{(c_1, c_2)\in \mathbb{F}_{2^m}\times \mathbb{F}_{2^m}|
h(c_1,c_2)=b,
c_1\neq c_2, c_i\neq 0, i=1, 2\}
\end{equation*}
and  $N(h(x_1,x_2)=b)=|S(h(x_1,x_2)=b)|$.

We consider two cases:

\medskip
 \noindent Case 1:  $m$ is odd, which implies that  Tr$_{2^m}(1)=1$.

   If $b=0$,   the number of solutions of the equation $h(x_1, x_2)=b$ is $2^m-v(b)=1$ by Lemma \ref{lem:solutions}.  Since $(0,0)$ is a solution,   $S(h(x_1, x_2)=b)=\varnothing$ and $N(h(x_1,x_2)=b)=0$.

If $b\neq 0$,  the number of solutions of the equation $h(x_1, x_2)=b$ is $2^m+1$ by Lemma \ref{lem:solutions}.  Thus,   $N(h(x_1,x_2)=b)=2^m+1-2$ since any element in $\mathbb{F}_{2^m}$ is a square.  We also obtain the corresponding $S(h(x_1,x_2)=b)$.

\medskip
\noindent Case 2: $m$ is even, which implies that  Tr$_{2^m}(1)=0$.

By Lemma \ref{lem:trace},  $x^2+x+1$ is reducible over $\mathbb{F}_{2^m}$  and so it can be written as  $x^2+x+1=(x+\alpha)(x+\beta)$ where $\alpha$, $\beta\in \mathbb{F}_{2^m}$, $\alpha\neq 1$, $\beta\neq 1$ and $\alpha\neq \beta$.  Thus,  $x^2_1+x_1x_2+x^2_{2}=(x_1+\alpha x_2)(x_1+\beta x_2)=b$ and so the number of solutions of $h(x_1, x_2)=b$ is $2^m+2^m-1$ if $b=0$ or  $2^m-1$ if $b\neq 0$.

 If $b=0$, then  $N(h(x_1,x_2)=b)=2^{m+1}-2$ and also we get $S(h(x_1,x_2)=b)$.

If $b\neq 0$,  $N(h(x_1,x_2)=b)=2^m-1-2$ since any element in $\mathbb{F}_{2^m}$ is a square.
We also get $S(h(x_1,x_2)=b)$.

From the above discussion, we get the following result.

\medskip
\begin{proposition} \label{prop: binary}
Let $\mathcal{G}$ be the Gabidulin code over $\mathbb{F}_{2^m}$ with $\mathbf{g}=\{g_1,\ldots, g_n\}$, dimension $k=1$ and $3\leq n\leq m$.   Let $f(x)=x^4+b x^2+cx$, where $b, c\in \mathbb{F}_{2^m}$.  Then $\sigma_f$ is not a deep hole of $\mathcal{G}$ if and only if there are two elements $\beta_{1}$ and $\beta_{2}$ in $<g_{1}, \cdots, g_{n} >$ such that $(\beta_{1}, \beta_{2}) \in S(h(x_1,x_2)=b)$. In particular, if $n=m$,  then $\sigma_f$ is  a deep hole of $\mathcal{G}$ if and only if $b=0$ and $m$ is odd.
\end{proposition}

\medskip

\begin{proof}  Note that two nonzero elements $\beta_{1}$ and $\beta_{2}$ are linearly independent over $\mathbb{F}_{2}$ if and only if $\beta_{1} \neq \beta_{2}$. Thus, by  Lemma   \ref{lem: deg of k+1}, $\sigma_f$ is not a deep hole of $\mathcal{G}$ if and only if there are  two distinct nonzero elements $\beta_{1}$ and $\beta_{2}$ in $<g_{1}, \cdots, g_{n} >$ such that
$$b=\beta^2_{1}+\beta_{1}\beta_{2}+\beta^2_{2},$$
          i.e., $(\beta_{1}, \beta_{2}) \in S(h(x_1,x_2)=b).$
In particular, if $n=m$, then $<g_{1}, \cdots, g_{n} >=\mathbb{F}_{2^m}$. By the above discussion, $\sigma_{f}$ is a deep hole only when $b=0$ and $m$ is odd. For the other cases, $N(h(x_1,x_2)=b)$ is at least $1$.  Therefore, the desired result is obtained.
\end{proof}

\begin{remark}  The second result of Proposition \ref{prop: binary} may not hold for the case of Hamming metric from Lemma \ref{lem: deg of k+1Ham} since it is possible that  $h(x_1,x_2)=b$ has no solutions in $\{g_1,\cdots,g_n\}$ when $b\neq 0$ although $h(x_1,x_2)=b$ always has solutions in $<g_{1}, \cdots, g_{n} >=\mathbb{F}_{2^m}$.
\end{remark}

\section{Conclusions}
In this paper, we study  deep holes of Gabidulin codes in both Hamming metric and rank metric. The general results for Hamming metric case (see Theorem \ref{thm: propertyH} and Lemma \ref{lem: deg of k+1Ham}) depend on the choice of the set $\{g_1, \ldots, g_n \}$, while the results for rank metric case (see Theorem \ref{thm: property} and Lemma \ref{lem: deg of k+1}) only depend on the subspace of $\mathbb{F}_{q^{m}}$ spanned by $g_1, \ldots, g_n$. In particular, when $n=m$, the latter  does not depend on the choice of $g_1, \ldots, g_n$ since $ < g_1, \ldots, g_n  >$ equals to the whole space $\mathbb{F}_{q^{m}}$.  Hence the problem about deep holes of Gabidulin codes in Hamming metric seems more complicated than in  rank metric.

On the other hand, for generalized Reed-Solomon codes,  it has been proved that the problem of determining if a received word is a deep hole is NP-hard \cite{Vardy}.   For Gabidulin codes, the problem  seems more complicated although we give a necessary and sufficient condition for this problem.  So we state it as  a conjecture.

\begin{conjecture}
Deciding deep holes of the Gabidulin code is NP-hard.
\end{conjecture}



\begin{thebibliography}{1}

\bibitem{Bartoli} D. Bartoli, M. Giulietti and I. Platoni, On the covering radius of MDS codes,
IEEE Trans. Inf. Theory 6 (2) (2015) 801-811.

\bibitem{Qi} Q. Cheng and E. Murray, On deciding deep holes of Reed-Solomon codes, Lecture notes in Computer Science 4484 (2007) 296-305.

\bibitem{QiWan} Q. Cheng and D. Wan, On the list and bounded distance decodability of
 Reed-Solomon codes, SIAM Journal on Computing 37 (1) (2007) 195-209.

\bibitem{Cohen85} G. Cohen, M. Karpovsky, H. Mattson and J. Schatz,  Covering radius--survey and recent results, IEEE Trans. Inf. Theory 31 (3) (1985) 328-343.


\bibitem{Cohen86} G. Cohen, A. C. Lobstein and  N. Sloane, Further results on the covering radius of codes,  IEEE Trans. Inf. Theory 32 (5) (1986) 680-694.

\bibitem{del} P. Delsarte,  Bilinear forms over a finite field with applications to coding theory,
J.  Comb. Theory, A 25 (3) (1978) 226-241.

\bibitem{gabidulin} E. M. Gabidulin, Theory of codes with maximum rank distance, Problemy
Peredachi Informatsii, 21 (1) (1985) 3-16.

\bibitem{Yan1} M. Gadouleau and Z. Yan, Packing and covering properties of rank
metric codes, IEEE Trans. Inf. Theory  54 (9) (2008) 3873-3883.


\bibitem{Yan2} M. Gadouleau and Z. Yan, Properties of codes with the rank metric,
in Proc. IEEE Globecom 2006, San Francisco, CA, 2006.


\bibitem{graham} R. Graham and N. Sloane, On the covering radius of codes, IEEE Trans. Inf. Theory 31 (3) (1985) 385-401.

\bibitem{Vardy} V. Guruswami and A. Vardy, Maximum-likelihood decoding of Reed-Solomon codes is NP-hard,  In Proceeding of SODA (2005) 2249-2256.

\bibitem{Tor} T. Helleseth, T. Klove and J. Mykkeltveit,  On the covering radius of binary codes, IEEE Trans. Inf. Theory  24 (5) (1978) 627-628.


\bibitem{Kuijper} A. Horlemann-Trautmann and M. Kuijper, Gabidulin decoding via minimal
bases of linearized polynomial modules, https://arxiv.org/abs/1408.2303v3.

\bibitem{Keti}M. Keti and D. Wan, Deep holes in Reed-Solomon codes based on Dickson polynomials,
Finite Fields  Appl. 40 (2016) 110-125.


\bibitem{Kotter} R. K\"{o}tter and R. R. Kschischang, Coding for errors and erasures in random networking coding,  IEEE Trans. Inf. Theory 54 (8) (2008) 3579-3591.

\bibitem{Liao} Q. Liao, On Reed-Solomon codes, Chinese Annals of Mathematics, (1) (2011) 89-98.

\bibitem{Nied} R. Lidl and H. Niederreiter, Finite fields, Cambridge University Press, Cambridge, London.

\bibitem{Loid}P. Loidreau, Properties of codes in rank metric, http://arxiv.org/abs/cs/0610057.


\bibitem{Silva} D. Silva, F. R. Kschischang and R. K\"{o}tter, A rank-metric approach to error control in random network coding,
IEEE Trans. Inf. Theory 54 (9) (2008) 3951-3967.

\bibitem{Vasa} W. B. Vasantha, N.  Suresh Babu, On the covering radius of rank-distance codes. Gaṇita Sandesh 13 (1) (1999) 43–48.

\bibitem{Wachter} A. Wachter-Zeh, Decoding of block and convolutional codes in rank metric. PhD thesis, Ulm University, Germany, 2013.


\bibitem{Wan2008}  D. Wan and Y. Li, On error distance of Reed-Solomon codes,  Science in China  51 (11) (2008) 1982-1988.


\bibitem{Wu} R. Wu and S. Hong, On deep holes of standard Reed-Solomon codes, Science
China Mathematics  55 (12) (2012) 2447-2455.

\bibitem{Zhang2013} J. Zhang, F.-W. Fu and Q. Liao, New deep holes of generalized Reed-Solomon codes, Scientia  Sinica 43 (7) (2013) 727-740.

\bibitem{ZW2016} J. Zhang and D. Wan, On deep holes of projective Reed-Solomon codes, International Symposium on Information Theory (2016) 925-929.

\bibitem{ZW2017} J. Zhang and D. Wan,  Explicit deep holes of Reed-Solomon codes, https://arxiv.org/abs/1711.02292.

\bibitem{Zhuang} J. Zhuang, Q. Cheng and J. Li, On determining deep holes of generalized Reed-Solomon codes, IEEE Trans. Inf. Theory 62 (1) (2016) 199-207.

\end{thebibliography}
\end{document}